\documentclass[lettersize,twocolumn]{article}

\usepackage[margin=2cm]{geometry}
\usepackage{cite}
\usepackage{graphicx}
\graphicspath{{../pdf/}{../jpeg/}}
\usepackage{amsmath, amssymb,amsfonts,amsthm}
\theoremstyle{definition}

\usepackage{algorithmic}
\usepackage{array}

\usepackage{physics}
\usepackage{tikz}
\usepackage{amsmath, amssymb, amsthm}
\usepackage{quantikz}
\usepackage{xcolor}
\usepackage{quantikz}
\usepackage{todonotes}
\everymath{\displaystyle}
\usepackage{graphicx} % Required for inserting images
\usepackage{caption}
\usepackage{subcaption}
\usepackage{authblk}
\usepackage{url} 
\usepackage{hyperref}
\usepackage{orcidlink}
\title{Exploiting all ancilla outcomes in linear combinations of unitaries: low-rank recovery and quantum trapdoor functions}

\author{Ammar Daskin \orcidlink{0000-0002-1497-5031}}
\affil{
Department of Computer Engineering\\
Istanbul Medeniyet University,
Istanbul, Turkiye, 34000\\
email: adaskin25@gmail.com
}
\date{}
\begin{document}

  \maketitle
\begin{abstract}
The linear combination of unitaries (LCU) is a fundamental quantum computing primitive used to embed non-unitary operators into larger systems via ancilla post-selection. In standard LCU implementations, only the zero-state ancilla outcome is retained, while all other measurement branches are discarded as unusable overhead. 
We study these discarded parts. Every computational-basis outcome of the ancilla projects the system onto a different linear combination of the target unitaries; collecting these outcome states in a matrix gives a factorization $\Phi=CX$, where $C$ is determined by the ancilla preparation and the rows of $X$ are the states $U_t\ket{\psi}$, so that $\operatorname{rank}(\Phi)\le K$. This holds for any LCU circuit. To make the analysis easier, we use a simplified model circuit with Hadamard gates on the index register and a single rotation qubit in place of the coefficient-preparation unitary, in which $C$ is explicit. The structure suggests two complementary applications:
(i) in an explicit amplitude-access model, classical low-rank matrix completion can reconstruct the full output (including the target) from a subset of its amplitude entries, while computational-basis counts provide intensities and generally require additional phase information for full complex recovery;
(ii) treating $C$ as a secret key suggests a candidate quantum trapdoor function, whose security remains informal and depends on the difficulty of structured phase retrieval.
The scheme thus turns the “junk” ancilla outcomes into a structured resource, possibly opening paths for further applications.

\end{abstract}
\textbf{Keywords:}
Linear combination of unitaries (LCU),  
Low-rank matrix completion,  
Quantum trapdoor functions,  
Matrix factorization,  
Quantum cryptography,  
Ancilla post-selection,  
Quantum circuit architecture

\section{Introduction}

Linear combinations of unitaries (LCU) are a fundamental technique in quantum algorithm design~\cite{childs2012hamiltonian}, because any matrix can be expressed as a sum of unitary operators. LCU can be implemented with a single ancilla register or via more flexible circuit ansätze~\cite{chakraborty2024implementing}, or as a programmable quantum circuits~\cite{daskin2012universal}. It has been used to implement Taylor‑series expansions of matrix functions~\cite{berry2015simulating} and structured matrices \cite{ollive2026hamiltonian} as well as linear combinations of permutation matrices, combined with Birkhoff–von Neumann decompositions and Sinkhorn’s algorithm~\cite{daskin2024quantum}.

Given \(K\) unitary operators \(U_1,\dots,U_K\) and complex coefficients \(\alpha_1,\dots,\alpha_K\), the goal is to apply the non‑unitary operator
\begin{equation}
T = \sum_{t=1}^K \alpha_t U_t
\label{eq:defT}
\end{equation}
to a quantum state \(\ket{\psi}\).
The standard LCU method~\cite{childs2012hamiltonian,low2019hamiltonian} uses an ancilla register of \(\lceil\log_2 K\rceil\) qubits and prepares the state
\begin{equation}
B\ket{0}_{\text{anc}} = \frac{1}{\sqrt{s}} \sum_{t=1}^K \sqrt{\alpha_t}\,\ket{t}_{\text{anc}},
\qquad s = \sum_{t=1}^K |\alpha_t|,
\end{equation}
followed by the controlled application of the \(U_t\) via the select operator
\(\text{Select}(U) = \sum_t \ket{t}\!\bra{t} \otimes U_t\),
and finally uncomputes the ancilla with \(B^\dagger\).
Post‑selecting on the ancilla returning to \(\ket{0}\) leaves the system in the unnormalized state \(\frac{1}{s}T\ket{\psi}\), with success probability
\begin{equation}
p_{\text{succ}}^{\text{(std)}} = \frac{\|T\ket{\psi}\|^2}{s^2}
= \frac{\|T\ket{\psi}\|^2}{\bigl(\sum_t|\alpha_t|\bigr)^2}.
\label{eq:stdprob}
\end{equation}

\paragraph{Motivation.}
In the standard LCU procedure~\cite{childs2012hamiltonian}, the ancilla register is returned to the \(\ket{0\cdots0}\) state by a unitary that encodes the coefficients; any other measurement outcome is treated as a failure and the corresponding system state is discarded.
While qubitization~\cite{low2019hamiltonian} and quantum singular value transformation (QSVT)~\cite{gilyen2019quantum} elegantly incorporate the entire block-encoding-including the “junk” subspace-into higher‑level algorithms by exploiting the full unitary, those discarded branches are seldom regarded as carriers of useful information in their own right.
However, in many LCU constructions the non-zero ancilla outcomes are mathematically related to the target operation (for instance, they correspond to the same linear combination with altered signs or coefficients) and therefore contain exploitable structure.
This raises a natural question: what do the junk outcomes contain, and can a simplified circuit make their structure explicit enough to be collectively processed, turning what is normally wasted into a resource?

\paragraph{Contribution.}
For any LCU circuit, the outcome states can be stacked into a matrix that factorizes as \(\Phi=CX\), where \(C\) is determined by the ancilla preparation. To make this structure explicit and easy to analyze, we use a simplified model circuit that replaces the fine‑tuned coefficient‑preparation unitary with Hadamard gates on an index register and a single rotation qubit.
The resulting unitary still block‑encodes the target operator \(T=\sum_{t=1}^K \alpha_t U_t\) (with \(|\alpha_t|\le 1\)), and every computational‑basis outcome of the index and rotation registers yields a distinct, unnormalized system state.
Collecting these states and reshaping them into a \(2K\times N\) matrix \(\Phi\) (where \(N=2^n\) for \(n\) system qubits) reveals the factorization \(\Phi = C X\).
Here, \(C\) is a known \(2K\times K\) coefficient matrix determined directly by the weights and the Hadamard signs, and the rows of \(X\) are the individual states \(U_t\ket{\psi}\).
The key algebraic observation is that \(\Phi\) has rank at most \(K\), despite its dimensions being much larger than the final target state. Consequently:
\begin{itemize}
    \item In an amplitude-access model, the complete \(2K\times N\) matrix \(\Phi\) can be recovered from a subset of its entries using classical low‑rank matrix completion~\cite{nguyen2019low}. With \(C\) known, recovery requires at least \(K\) suitably chosen observed entries per column. Computational-basis measurements instead estimate intensities, so full complex recovery generally requires additional measurement settings or phase information.
\item We describe an informal information-hiding idea in which the coefficient matrix \(C\) acts as a secret key. Non-uniqueness of the factorization motivates a candidate quantum trapdoor function, but does not by itself prove hiding or one-wayness. Moreover, the overall circuit satisfies \(V^2 \propto \bigoplus_t U_t^2\) (the coefficients cancel), suggesting an involutory symmetric-circuit construction; by itself this gives no confidentiality.
\end{itemize}

In both applications, the “junk’’ ancilla outcomes are not discarded; instead they form the essential algebraic structure of the enlarged output matrix, enabling structured recovery and motivating a candidate keyed primitive within a single circuit design.
In the following sections we detail the circuit, its algebraic properties, and the two application avenues.

\section{The circuit and its output}
\label{sec:circuit}

\begin{figure*}[htpb]
\centering
\begin{quantikz}
  % Index ancilla 0 (least significant)
\lstick{$\ket{0}_{idx0}$}
  & \gate{H} & \octrl{1}\gategroup[4,steps=4,style={dashed, rounded corners,fill=blue!10, inner xsep=2pt},background,label style={label position=below,anchor=north,yshift=-0.3cm}]{Select: $R_t \otimes U_t$} & \octrl{1} & \ctrl{1} &\ctrl{1} &\gate{H} & \meter{} \\
  % Index ancilla 1 (most significant)
  \lstick{$\ket{0}_{idx1}$}
  & \gate{H} & \octrl{1} &\ctrl{1}&\octrl{1}&\ctrl{1}& \gate{H} & \meter{} \\
  % Rotation qubit
 \lstick{$\ket{0}_{rot}$}
  & \qw & 
\gate[2,style={fill=green!15}]{R_0\otimes U_0}&
\gate[2,style={fill=green!15}]{R_1\otimes U_1}&
\gate[2,style={fill=green!15}]{R_2\otimes U_2}&
\gate[2,style={fill=green!15}]{R_3\otimes U_3}&\qw& \\
  % System register (n qubits)
  \lstick{$\ket{\psi}_{sys}$} 
  & \qwbundle{n} \qw & \qw & \qw & \ &  & \qw & \qw
\end{quantikz}
\caption{Simplified model circuit for LCU with $K=4$ terms. The index register is placed in an equal superposition by Hadamard gates, then controls the application of $R_t\otimes U_t$ on the rotation and system registers. After the controlled operations, a second set of Hadamard gates mixes the index states. Measurement of all ancilla qubits yields a collection of outcomes, each carrying a linear combination of the unitaries.}
\label{fig:circuit}
\end{figure*}

We consider the task of implementing the following linear combination inside of a larger unitary circuit:
\begin{equation}
T = \sum_{t=1}^{K} \alpha_t U_t , \qquad \alpha_t\in\mathbb{R},
\label{eq:defT_model}
\end{equation}
where each $U_t$ is an $n$-qubit unitary and the coefficients are real and scaled so that $|\alpha_t|\le 1$ (note that complex coefficients can be handled by absorbing phases into the $U_t$ or by using two rotation qubits in our circuit design). While the standard LCU circuit where the coefficients implemented through structured or specific matrices such as Householder transformations result in structured special matrices, to make the analysis and circuit design easier, here we employ a circuit construction which can be considered generalization of circuit design approaches  used to do matrix arithmetic between unitary matrices by using Hadamard gates. This type of constructions and many of known block encoding ideas where the structure of Hadamard mixing used along with rotation gates (e.g. see \cite{chakraborty2024implementing,daskin2012universal} for cascading construction of unitary blocks) makes the design of the overall circuit easier: As we will show here after a simple qubit reordering, the resulting unitary $U$ which represents the circuit in Fig.\ref{fig:circuit} acquires a $2\times2$ block structure
$\begin{pmatrix} A & B \\ -B & A \end{pmatrix}$ (or a reflection variant),
where the diagonal blocks of $A$ are identical and equal to
$D_w/K = \frac1K\sum_{t=1}^{K} \alpha_t U_t$,
and those of $B$ equal $D_r/K = \frac1K\sum_{t=1}^{K} \sqrt{1-\alpha_t^2} U_t$.

This circuit design (depicted in Fig.\ref{fig:circuit} for K=4) consists of three registers:
\begin{enumerate}
    \item A register of $\lceil\log_2 K\rceil$ index qubits initialised with Hadamard gates;
    \item A single rotation qubit that carries two‑channel rotation gates, each controlled on the index;
    \item The target unitaries $U_t$, controlled on the index, as in standard LCU.
\end{enumerate}
Then the following sequence of operations is applied to thes registers:
\begin{itemize}
\item Hadamard gates on all index qubits, creating an equal superposition $|+\rangle_{\!I} = \frac{1}{\sqrt{K}}\sum_{t=1}^{K} |t\rangle$ (we label the computational basis states $|t\rangle$ with $t=1,\dots,K$).
\item A controlled block which, conditionally on the index being $|t\rangle$, applies the two-register gate $R_t\otimes U_t$ to the rotation qubit and the system. Here, for each term $t$ we define a single‑qubit rotation
\begin{equation}
R_t = \begin{pmatrix} w_t & r_t \\[2pt] r_t & -w_t \end{pmatrix},
\qquad w_t^2+r_t^2=1,
\label{eq:Rt}
\end{equation}
which we call the \emph{reflection variant}. (The cyclic variant where negative sign on an off-diagonal entry leads to a similar structure up to a sign difference; all other considerations apply analogously.)  
Here, the rotation gate angles are set to values so that the top–left entry of $R_t$ equals the desired coefficient. That means, 
\begin{equation}
w_t = \alpha_t,\qquad r_t = \sqrt{1-\alpha_t^2}.
\label{eq:walpha}
\end{equation}

\item A second round of Hadamard gates on the index register.
\item Measurement of the index and rotation qubits in the computational basis.
\end{itemize}

\subsection{The application of operations and output state}
For an arbitrary input, $\ket{\psi}$, the initial state is $|0\rangle_{\!I}^{\otimes m}|0\rangle_{\!R}|\psi\rangle_{\!S}$.   After the first Hadamard layer the index register becomes $|+\rangle_{\!I}$, and the whole state is
\begin{equation}
\frac{1}{\sqrt{K}} \sum_{t=1}^{K} |t\rangle_{\!I} \, |0\rangle_{\!R} |\psi\rangle_{\!S}.
\end{equation}

We can describe the following controlled application of the $R_t\otimes U_t$ gates by a block diagonal operator:
\begin{equation}
M = \bigoplus_{t=1}^{K} \bigl( R_t \otimes U_t \bigr),
\label{eq:M}
\end{equation}
which acts on the rotation and system registers only when the index is $|t\rangle$. 
Applying $M$ yields
\begin{equation}
\frac{1}{\sqrt{K}} \sum_{t=1}^{K} |t\rangle_{\!I} \; \bigl(R_t\otimes U_t\bigr) |0\rangle_{\!R}|\psi\rangle_{\!S}.
\end{equation}
The final pair of Hadamard gates on the index register, described by the normalized Hadamard matrix $H$ with entries $H_{it}\in\{+1/\sqrt{K},-1/\sqrt{K}\}$, transforms this state into the following final state:
\begin{equation}
|\Psi_{\text{out}}\rangle = 
\sum_{i=1}^{K} |i\rangle_{\!I} \;
\Biggl[ \frac{1}{\sqrt{K}} \sum_{t=1}^{K} H_{it} \bigl(R_t\otimes U_t\bigr) |0\rangle_{\!R}|\psi\rangle_{\!S} \Biggr].
\label{eq:PsiOut}
\end{equation}

Now we measure the index and rotation qubits.  Conditioned on obtaining the outcome $i$ for the index and $r\in\{0,1\}$ for the rotation, the (unnormalized) system state becomes
\begin{equation}
|\phi_{i,r}\rangle = \frac{1}{\sqrt{K}} \sum_{t=1}^{K} H_{it} \, \langle r| R_t |0\rangle \, U_t |\psi\rangle.
\label{eq:phii0}
\end{equation}
Using $\langle 0|R_t|0\rangle = w_t$ and $\langle 1|R_t|0\rangle = r_t$, and introducing the sign
\begin{equation}
s_{i,t} = \sqrt{K}\, H_{it} \;\in\; \{+1,-1\},
\label{eq:sign}
\end{equation}
we obtain the compact expressions
\begin{align}
|\phi_{i,0}\rangle &= \frac{1}{K} \sum_{t=1}^{K} s_{i,t}\, w_t\, U_t |\psi\rangle, \label{eq:phi_i0}\\
|\phi_{i,1}\rangle &= \frac{1}{K} \sum_{t=1}^{K} s_{i,t}\, r_t\, U_t |\psi\rangle. \label{eq:phi_i1}
\end{align}
For the all‑plus index outcome $i=0$ we have $s_{0,t}=+1$ for all $t$, so we can write
\begin{equation}
|\phi_{0,0}\rangle = \frac{1}{K} D_w |\psi\rangle,\qquad
|\phi_{0,1}\rangle = \frac{1}{K} D_r |\psi\rangle,
\label{eq:phi00}
\end{equation}
where $D_w := \sum_t w_t U_t = T$ and $D_r := \sum_t r_t U_t$. Here also note that $r_t=\sqrt{1-w_t^2}$.  From the above equations, we can describe the probability of a specific outcome $(i,r)$ as:
\begin{equation}
p_{i,r} = \bigl\| |\phi_{i,r}\rangle \bigr\|^2.
\end{equation}

Because the rotation gate only allows $|w_t|\le 1$, in the case we have large or all coefficients close to zero, we can globally scale the desired coefficients:
\begin{equation}
c = \max_{t} |\alpha_t|,\qquad \beta_t = \frac{\alpha_t}{c},
\end{equation}
and set $w_t = \beta_t$.  Then $D_w = \sum_t \beta_t U_t = T/c$, and the outcome $(0,0)$ gives the unnormalized state
\begin{equation}
|\phi_{0,0}\rangle = \frac{1}{Kc}\, T|\psi\rangle.
\label{eq:scaled_T}
\end{equation}
Therefore, we can describe its success probability from the norm of this state (the chance of measuring only that specific outcome) as:
\begin{equation}
p_{0,0} = \frac{\|T|\psi\rangle\|^2}{c^{2} K^{2}}.
\label{eq:p00}
\end{equation}
This success probability is smaller than the standard LCU value $\|T\ket{\psi}\|^2/(\sum_t|\alpha_t|)^2$ unless the coefficients are all $O(1)$. The purpose of the model circuit is not to improve this post-selection probability, but to make all outcome combinations explicit.

The circuit also generates the combination \(\frac{1}{K}\sum_t r_t U_t\) when the rotation qubit is measured in \(\ket{1}\).  
One may choose to post‑select only on the index register being \(\ket{0}\), irrespective of the rotation outcome; this yields an ensemble containing both \(D_w\ket{\psi}\) and \(D_r\ket{\psi}\), and the total probability of obtaining index \(\ket{0}\) is exactly  $p_{0,0}+p_{0,1}=(\|D_w\ket{\psi}\|^2+\|D_r\ket{\psi}\|^2)/K^2$. These probabilities along with the probabilities of the standard LCU are compared in Fig.\ref{fig:probability-comparison}. The figure compares the success probabilities of standard LCU and the simplified model circuit for $K=8$, $N=64$, and a family of coefficient vectors where half the $\alpha_t$ are fixed to $1$ and the other half are swept from $0.1$ to $1$.
The curves labelled ``With rotation qubit'' correspond to the probability $p_{0,0}$ of obtaining the single outcome $(i=0,r=0)$, i.e. the strict equivalent of standard post‑selection.
``Without rotation qubit'' gives the probability $p_{0,0}+p_{0,1}$ of measuring the index register in $\ket{0}$ irrespective of the rotation qubit.
Analytic predictions (dashed) align perfectly with the simulated data, confirming the formulas derived in the text.

\begin{figure}[htbp]
    \centering
    \includegraphics[width=\linewidth]{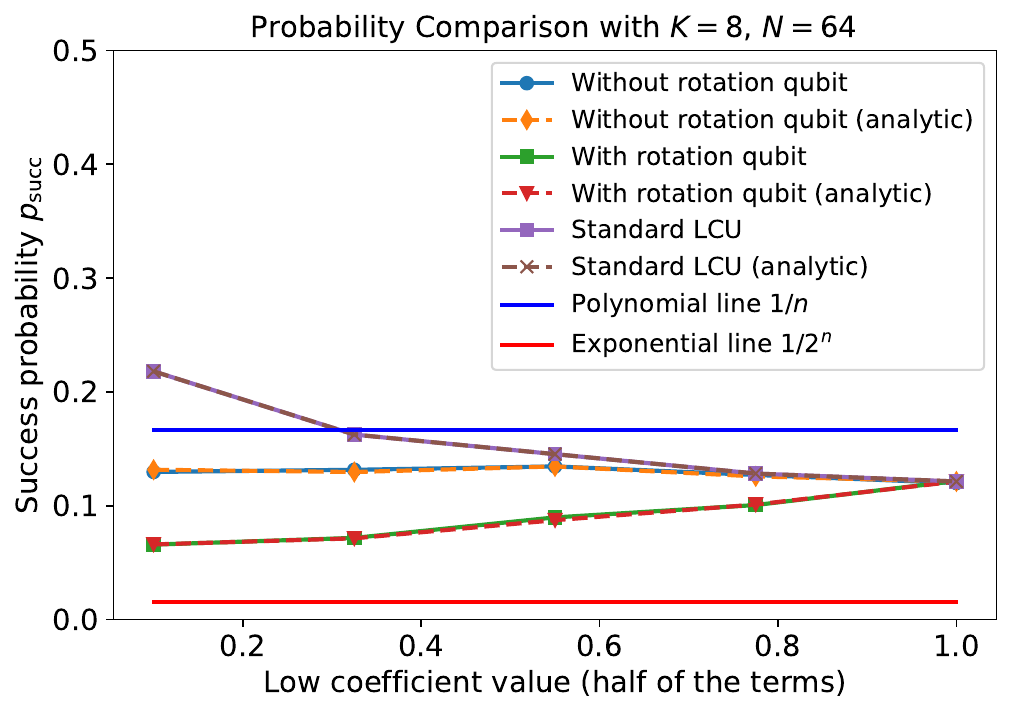}
    \caption{
    Success probabilities for \(K=8\), \(N=64\) with half the coefficients fixed to \(1\) and the other half varying from \(0.1\) to \(1\).
    The solid lines show simulation results; dashed lines are the analytic predictions.
    ``With rotation qubit" corresponds to the strict post‑selection on \(i=0,r=0\) (probability \(p_{0,0}\)).
    ``Without rotation qubit" sums both rotation outcomes for index \(i=0\) (probability \(p_{0,0}+p_{0,1}\)).
    Standard LCU post‑selection is shown for reference.
    }
    \label{fig:probability-comparison}
\end{figure}
\section{Unitary matrix representation of the circuit and its properties}
\label{sec:matrix_repr}

Although we will work directly with the output states~\eqref{eq:phi_i0}--\eqref{eq:phi_i1}, it is illuminating to view the full circuit as a block‑encoding matrix. We can describe the whole circuit operation before measurement as the following unitary product:
\begin{equation}
V \;=\; (H\otimes I_2\otimes I_N)\;M\;(H\otimes I_2\otimes I_N),
\label{eq:Vdef}
\end{equation}
where $H$ is the normalized $K\times K$ Hadamard matrix and
\begin{equation}
M \;=\; \bigoplus_{t=1}^{K}\,\bigl(R_t\otimes U_t\bigr)
\label{eq:Mdef}
\end{equation}
is block‑diagonal in the index register.

Let $\ket{i}$ ($i=1,\dots,K$) denote the computational basis of the index register.
The $(i,j)$ block of $V$ acting on the rotation$\otimes$system subspace is the $(2N)\times(2N)$ matrix
\begin{equation}
\begin{split}
V_{ij} & \;=\; \sum_{t=1}^{K} H_{it}H_{jt}\;(R_t\otimes U_t)\\[2pt]
& \;=\; 
\begin{pmatrix}
\sum_{t} H_{it}H_{jt}\,w_tU_t &
\sum_{t} H_{it}H_{jt}\,r_tU_t \\[14pt]
\sum_{t} H_{it}H_{jt}\,r_tU_t &
-\sum_{t} H_{it}H_{jt}\,w_tU_t
\end{pmatrix}.
\end{split}
\label{eq:Vij}
\end{equation}
For $i=j$, using $H_{it}^2=1/K$, every diagonal block $V_{ii}$ is identical and encodes the desired linear combinations:
\begin{equation}
V_{ii} \;=\; \frac{1}{K}
\begin{pmatrix}
\sum_t w_tU_t & \sum_t r_tU_t \\[8pt]
\sum_t r_tU_t & -\sum_t w_tU_t
\end{pmatrix}
\;=\; \frac{1}{K}
\begin{pmatrix}
D_w & D_r \\[2pt]
D_r & -D_w
\end{pmatrix}.
\label{eq:Vdiag}
\end{equation}
For $i\neq j$, the products $H_{it}H_{jt}$ equal $\pm1/K$, yielding
\begin{equation}
V_{ij} \;=\;
\begin{pmatrix}
\displaystyle\sum_{t} H_{it}H_{jt}\,w_tU_t &
\displaystyle\sum_{t} H_{it}H_{jt}\,r_tU_t \\[14pt]
\displaystyle\sum_{t} H_{it}H_{jt}\,r_tU_t &
-\displaystyle\sum_{t} H_{it}H_{jt}\,w_tU_t
\end{pmatrix}.
\label{eq:Voffdiag}
\end{equation}

\subsubsection{Shuffling to the global $2\times2$ block form}

Reordering the basis from $|\text{index}\rangle\otimes|\text{rotation}\rangle\otimes|\text{system}\rangle$ to $|\text{rotation}\rangle\otimes|\text{index}\rangle\otimes|\text{system}\rangle$ (via a permutation $S$) exposes the $2\times2$ structure.
Let $S$ be the $(2NK)\times(2NK)$ permutation matrix (a sequence of swap operations) which performs this reordering so that the final unitary becomes:
\begin{equation}
U \;=\; S\;V\;S^{\mathsf T}
\;=\; \begin{pmatrix}
A & B \\[2pt]
B & -A
\end{pmatrix},
\label{eq:Ushuffled}
\end{equation}
This representation is possible because $A$ and $B$ are each $(NK)\times(NK)$ matrices built from $K^2$ blocks of
size $N\times N$:
\begin{equation}
\begin{split}
&A_{ij} \;=\; \sum_{t=0}^{K-1} H_{it}H_{jt}\;w_tU_t,\\
&B_{ij} \;=\; \sum_{t=0}^{K-1} H_{it}H_{jt}\;r_tU_t.
\end{split}
\label{eq:ABblocks}
\end{equation}
Therefore, this matrix is just column and row permuted version of the original circuit (the registers are swapped).
This representation reveals the diagonal sub‑blocks ($i=j$) which are the desired parts of the circuit:
\begin{equation}
\begin{split}
&A_{ii} \;=\; \frac{1}{K}\,D_w \;=\; \frac{1}{K}\sum_{t=0}^{K-1} w_tU_t,\\
&B_{ii} \;=\; \frac{1}{K}\,D_r \;=\; \frac{1}{K}\sum_{t=0}^{K-1} r_tU_t.
\end{split}
\label{eq:ABdiag}
\end{equation}
All diagonal blocks are identical; each \(A_{ii}\) encodes \(D_w/K\) and each \(B_{ii}\) encodes \(D_r/K\).
Consequently, unlike standard LCU where one must post‑select on the ancilla being \(\ket{0\cdots0}\), here every computational basis state of the index register—paired with the rotation qubit in \(\ket{0}\) (resp. \(\ket{1}\))—projects onto the same linear combination \(D_w/K\) (resp. \(D_r/K\)). The output states~\eqref{eq:phi_i0} and \eqref{eq:phi_i1} are exactly these operators applied to \(\ket{\psi}\) and weighted by the Hadamard sign pattern.

Note that many combinations of the linear combinations are possible by using the matrix entries. For instance, one can see that powers of this matrix has resemblance to the descriptions in qubitization circuits \cite{low2019hamiltonian} where the iterated dynamics can be expressed through Chebyshev polynomials of \(A\) and \(B\).

\subsection{Eigenvalues and singular values from the explicit block structure}
\label{sec:CSD}

The shuffled unitary~\eqref{eq:Ushuffled},
has a block structure that is completely determined by the circuit parameters. We can write
\begin{equation}
\begin{split}
A  & = \sum_{i,j} |i\rangle\langle j| \otimes \Bigl(\sum_t H_{it}H_{jt}\, w_t U_t\Bigr)\\
    &= (H \otimes I_N)\;
      \begin{pmatrix}
        w_1 U_1 & & \\
        & \ddots & \\
        & & w_K U_K
      \end{pmatrix}
      \;(H^{\! \top} \!\otimes I_N)
      \\
&\;=\; Q \, W \, Q^{\dagger},
\end{split}
\label{eq:A_similarity}
\end{equation}
where \(Q = H \otimes I_N\) is unitary (the normalized Hadamard satisfies \(H^{\!\top}H = I_K\)) and
\begin{equation}
W = \operatorname{diag}\bigl(w_1 U_1,\; w_2 U_2,\; \dots,\; w_K U_K\bigr).
\label{eq:Wdef}
\end{equation}
Thus \(A\) is unitarily similar to a block‑diagonal matrix whose \(t\)-th block is \(w_t U_t\).  This gives the eigenvalues of \(A\) as the union of the eigenvalues of \(w_t U_t\) for \(t=1,\dots,K\).  Because each \(U_t\) is unitary, its eigenvalues lie on the unit circle; consequently the eigenvalues of \(A\) lie on circles of radius \(|w_t|\) and its norm is \(\|A\| = \max_t |w_t|\).

An identical argument can be applied to \(B\) to obtain
\begin{equation}
B = Q \, \operatorname{diag}(r_1 U_1,\dots,r_K U_K) \, Q^{\dagger}.
\end{equation}
And so similarly the eigenvalues of \(B\) lie on circles of radius \(|r_t| = \sqrt{1-|w_t|^2}\).

Similarly, because \(w_t U_t\) has all singular values equal to \(|w_t|\) (with multiplicity \(N\)), the multiset of singular values of \(A\) consists of \(KN\) numbers, namely
\begin{equation}
\underbrace{|w_1|,\dots,|w_1|}_{N},\; \underbrace{|w_2|,\dots,|w_2|}_{N},\; \dots,\; \underbrace{|w_K|,\dots,|w_K|}_{N}.
\label{eq:svdA}
\end{equation}
The same holds for \(B\) with \(|r_t| = \sqrt{1-|w_t|^2}\).

The unitarity of \(U\) along with above singular value decompositon of $A$ and $B$ properties imposes the relations to write the cosine‑sine decomposition (CSD) of matrix $U$ (Note that CSD of a unitary can be used to design quantum circuits \cite{mottonen2004quantum,nakajima2009synthesis,chen2013qcompiler}) : We write each unitary in its polar decomposition \(U_t = P_t Q_t^{\dagger}\) with \(P_t, Q_t\) unitaries; then
\begin{equation}
w_t U_t = P_t\, (w_t I_N)\, Q_t^{\dagger}, \qquad
r_t U_t = P_t\, (r_t I_N)\, Q_t^{\dagger}.
\end{equation}
Define the block‑diagonal matrices
\begin{equation}
\mathcal{P} = \operatorname{diag}(P_1,\dots,P_K), \qquad
\mathcal{Q} = \operatorname{diag}(Q_1,\dots,Q_K),
\end{equation}
and let \(\Sigma_w = \operatorname{diag}(w_1 I_N,\dots,w_K I_N)\), \(\Sigma_r = \operatorname{diag}(r_1 I_N,\dots,r_K I_N)\).  Then
\begin{equation}
W = \mathcal{P} \, \Sigma_w \, \mathcal{Q}^{\dagger}, \qquad
\operatorname{diag}(r_t U_t) = \mathcal{P} \, \Sigma_r \, \mathcal{Q}^{\dagger}.
\end{equation}
Substituting into the expressions for \(A\) and \(B\) yields the standard form
\begin{equation}
A = Q_1 \, \Sigma_w \, Q_2^{\dagger}, \qquad
B = Q_1 \, \Sigma_r \, Q_2^{\dagger},
\label{eq:explicitCSD}
\end{equation}
with the unitary matrices
\begin{equation}
Q_1 = (H \otimes I_N) \, \mathcal{P}, \qquad
Q_2 = (H \otimes I_N) \, \mathcal{Q}.
\label{eq:Q1Q2}
\end{equation}
The diagonal entries of \(\Sigma_w\) and \(\Sigma_r\) are the singular values of \(A\) and \(B\); they automatically satisfy \(\sigma_{w,j}^2 + \sigma_{r,j}^2 = 1\) because \(w_t^2 + r_t^2 = 1\) for every term.

Inserting the CSD~\eqref{eq:explicitCSD} into \(U\) gives
\[
U = (Q_1 \otimes I_2)\,
\begin{pmatrix}
\Sigma_w & \Sigma_r \\[2pt]
\Sigma_r & -\Sigma_w
\end{pmatrix}\,
(Q_2 \otimes I_2)^{\dagger}.
\]
The central matrix is permutation‑similar to a direct sum of \(2\times2\) blocks
\begin{equation}
M_j = \begin{pmatrix}
\sigma_{w,j} & \sigma_{r,j} \\[2pt]
\sigma_{r,j} & -\sigma_{w,j}
\end{pmatrix}, \qquad j = 1,\dots,NK.
\label{eq:Mj}
\end{equation}
Each \(M_j\) satisfies \(\det(M_j) = -(\sigma_{w,j}^{2}+\sigma_{r,j}^{2}) = -1\) and \(\operatorname{tr}(M_j)=0\), so the central CSD matrix has eigenvalues \(\pm1\). This does not by itself determine the spectrum of \(U\), because \(Q_1\) and \(Q_2\) need not be equal. Instead, \(V=HMH\) is similar to \(M=\bigoplus_t(R_t\otimes U_t)\), and each \(R_t\) has eigenvalues \(\pm1\). Hence
\[
\operatorname{spec}(V)=\bigcup_t\{\pm\mu:\mu\in\operatorname{spec}(U_t)\},
\]
and the shuffled unitary \(U\), which is permutation-similar to \(V\), has the same spectrum. Its eigenvalues are \(\pm1\) only when all \(U_t^2=I\). Since \(R_t^2=I\),
\[
V^2=(H\otimes I)\left[\bigoplus_t(I_2\otimes U_t^2)\right](H\otimes I),
\]
so the coefficients cancel in \(V^2\), and \(V\) is an involution when all \(U_t^2=I\).

\subsubsection{Connection to the success probability}
Since the diagonal blocks of \(U\) that contain the desired linear combination correspond to the subspace where the rotation qubit is \(|0\rangle\) and the index register is in a particular computational basis state; their collective weight is governed by the singular values of \(A\), which are exactly the scaled coefficients \(|w_t|\).  Therefore, when all \(|w_t|\) are small, the singular values of \(A\) are concentrated near \(0\) (and those of \(B\) near \(1\)), so the probability of the ``good'' ancilla outcome is tiny.  

\section{Low‑rank structure and recovery of the full output state}
\label{sec:recovery}

In standard LCU implementations, only the post‑selected outcome \(\ket{0\cdots0}\) is kept; all other shots are discarded.
In our scheme, by contrast, \emph{every} computational‑basis measurement of the ancilla—index and rotation—projects the system onto a distinct linear combination of the unitaries, yielding an unnormalized state \(\ket{\phi_{i,r}}\) from Eqs.~\eqref{eq:phi_i0}--\eqref{eq:phi_i1}.
Collecting all these outcomes and reshaping them into a \(2K\times N\) matrix \(\Phi\) reveals the algebraic structure which we can exploit. 

\subsection{Factorization and rank bound}

After running the circuit of Fig.\ref{fig:circuit} and measuring the index, rotation, and system qubits, one obtains samples from the $2K$ unnormalized states $\ket{\phi_{i,r}}$ of Eqs.~\eqref{eq:phi_i0}--\eqref{eq:phi_i1}.
For $r\in\{0,1\}$, $i\in\{1,\dots,K\}$, and $k\in\{1,\dots,N\}$; collecting these states as row vectors yields the full output matrix
\begin{equation}
\Phi_{(r,i),\,k} = \langle k | \phi_{i,r}\rangle.
\label{eq:Phi_def}
\end{equation}
From this relation and using the algebraic structure of the circuit, we can write the following factorization:
\begin{equation}
\Phi = C \, X,
\label{eq:factorization}
\end{equation}
where $C$ is the $2K\times K$ coefficient matrix
\begin{equation}
C_{(0,i),t} = \frac{1}{K}\, s_{i,t}\, w_t, \qquad
C_{(1,i),t} = \frac{1}{K}\, s_{i,t}\, r_t,
\label{eq:Cdef}
\end{equation}
with $s_{i,t} = \sqrt{K} H_{it} \in \{+1,-1\}$, and $X$ is the $K\times N$ matrix whose $t$‑th row is the state $\langle k|U_t|\psi\rangle$.

Because $C$ has full column rank $K$, Eq.~\eqref{eq:factorization} immediately implies
\begin{equation}
\operatorname{rank}(\Phi) \le \min(K,\operatorname{rank}(X)) \le K.
\label{eq:rank_bound}
\end{equation}
Hence $\Phi$ is a low‑rank matrix of size $2K\times N$ with rank at most $K$, a direct consequence of the circuit generating only $K$ independent linear combinations of the unitaries.

\subsection{Measurement challenge}

A practical quantum measurement in the computational basis estimates the probability $p_{r,i,k} = |\Phi_{(r,i),k}|^2$, i.e., the squared magnitude.
To recover the full complex matrix $\Phi$ one needs not only magnitudes but also phases.
Several strategies can be pursued:
\begin{itemize}
\item \textbf{Multi‑basis measurements.} Performing additional measurements in, e.g., the Pauli $X$ and $Y$ bases allows full quantum state tomography on the $2K$ unnormalized states~\cite{cramer2010efficient,gross2010quantum}. This recovers the complex amplitudes directly, but the shot cost grows with the number of bases.
\item \textbf{Real‑amplitude restriction.} If the coefficients $w_t$ and the unitaries $U_t$ are chosen such that all amplitudes become real and non‑negative (e.g., using the cyclic rotation variant and permutation matrices), the square root of the observed frequency directly gives the amplitude up to known signs from $C$. The low‑rank recovery then operates on a real matrix, greatly simplifying the task.
\item \textbf{Real‑imaginary decomposition.} For general complex amplitudes one can write $\Phi = \Phi_{\rm re} + i\,\Phi_{\rm im}$. Both $\Phi_{\rm re}$ and $\Phi_{\rm im}$ are real matrices of rank $\le K$ (since $C$ is real and $X$ splits into real and imaginary parts). One can apply matrix completion algorithms to $\Phi_{\rm re}$ and $\Phi_{\rm im}$ separately, each using only the squared‑magnitude constraints $|\Phi_{(r,i),k}|^2 = (\Phi_{\rm re})_{(r,i),k}^2 + (\Phi_{\rm im})_{(r,i),k}^2$ as side information. This opens the door to phase‑retrieval‑style reconstruction while exploiting the low‑rank structure.
\end{itemize}
In what follows we present recovery algorithms that work on the complex amplitude matrix (or its real/imaginary parts) assuming that enough entries (or their projections) have been estimated through one of the above routes.

\subsection{Recovery algorithms}

Given a subset of observed entries $\Omega \subset \{1,\dots,2K\}\times\{1,\dots,N\}$ and noisy estimates $\Phi^{\rm obs}_{ij}$ for $(i,j)\in\Omega$, the task is to recover a rank‑$K$ matrix that best fits the data.
We outline two complementary families of algorithms \cite{nguyen2019low,davenport2016overview}.

\paragraph{Generic low‑rank matrix completion.}
These methods treat $\Phi$ as an arbitrary low‑rank matrix and do not assume knowledge of $C$.

\begin{itemize}
\item \textbf{Singular Value Projection (SVP) / Iterative Hard Thresholding (IHT)}~\cite{jain2010guaranteed}.
Starting from an initial guess (e.g., zero), one alternates a gradient descent step on the observed entries,
\[
\Phi^{(\ell+1/2)} = \Phi^{(\ell)} + \mu \, \Pi_\Omega(\Phi^{\rm obs} - \Phi^{(\ell)}),
\]
where $\Pi_\Omega$ keeps only the observed positions and $\mu$ is a step size, with a rank‑$K$ truncation via the singular value decomposition:
\[
\Phi^{(\ell+1)} = \mathcal{T}_K\bigl(\Phi^{(\ell+1/2)}\bigr) = \sum_{j=1}^K \sigma_j\, \mathbf{u}_j \mathbf{v}_j^\dagger.
\]
The method is simple and provably convergent under incoherence assumptions~\cite{candes2012exact}. Similarly, an alternate known as nuclear norm minimization~\cite{cai2010singular} can be also used to form the matrix with minimal nuclear norm (sum of singular values) consistent with the observed entries. In this case, the convex relaxation is solved by semidefinite programming or first‑order methods, and similarly gives theoretical guarantees when the observation pattern is incoherent.

\item \textbf{Alternating Least Squares (ALS)}~\cite{koren2009matrix}.
The matrix is parametrized as $A B^\dagger$ with $A\in\mathbb{C}^{2K\times K}$, $B\in\mathbb{C}^{N\times K}$, and the factors are optimized alternately by solving least‑squares problems on the observed entries. ALS often converges faster than SVP and is widely used in recommendation systems, though it requires careful initialization.
\end{itemize}

Generic completion guarantees depend on the observation pattern and standard incoherence assumptions. The structured left factor $C$ can affect performance, so the factorized approach below uses this known structure directly rather than relying only on generic low-rank recovery.

\paragraph{Factorized recovery using the known $C$.}
When the coefficient matrix $C$ is known (which is the case once the circuit parameters $\alpha_t$ are fixed), the factorization~(\ref{eq:factorization}) can be used to define the low‑rank constraint directly.
One then recovers the smaller matrix $X\in\mathbb{C}^{K\times N}$ by solving
\begin{equation}
\min_{X} \sum_{(i,j)\in\Omega} \bigl| (C X)_{ij} - \Phi^{\rm obs}_{ij} \bigr|^2,
\label{eq:factorized_ls}
\end{equation}
which is a simple linear least‑squares problem.
Because $X$ has only $K\times N$ unknowns (compared to $2K\times N$ in $\Phi$), the number of required observations per column drops to $K$; with random observation masks, the problem becomes well‑posed as soon as each column of $\Phi$ features at least $K$ observed rows.
The solution can be obtained either column‑wise via regularized normal equations,
\[
(C_{\text{obs}}^\dagger C_{\text{obs}} + \lambda I) \, X_{\cdot,j} = C_{\text{obs}}^\dagger \, \Phi^{\rm obs}_{\cdot,j},
\]
or jointly via gradient descent on $X$.
This factorized approach is exact in the noiseless case when every column has at least $K$ observed rows whose corresponding rows of $C$ are linearly independent, and it bypasses the generic incoherence requirement. With a random mask, however, a substantial observation fraction may be needed before this condition holds for every column.
Once $X$ is recovered, the target vector follows as
\begin{equation}
T\ket{\psi} = \sum_{t=1}^K \alpha_t \, \bigl(\text{row }t\text{ of }X\bigr).
\label{eq:extract_target}
\end{equation}

\subsection{Numerical illustration}

Fig.\ref{fig:recovery_plot} shows the relative error in reconstructing the full output matrix \(\Phi\) and the target row \(\varphi\) (the unnormalized \(T|\psi\rangle/K\)) as a function of the fraction of randomly observed entries, using the exact amplitudes computed via the matrix representation of Sec.~\ref{sec:circuit}. These experiments therefore belong to the amplitude-access model described above; they are not simulations of computational-basis measurement counts.  
We compare the generic Singular Value Projection (SVP) completion~\cite{jain2010guaranteed} with the factorized recovery of Sec.~\ref{sec:recovery} that exploits the known \(C\) and solves Eq.~\eqref{eq:factorized_ls} column by column.

\paragraph{Observation requirements.}
The factorized model \(\Phi = C X\) has only \(KN\) unknown parameters (the entries of \(X\)), and each observed \(\Phi_{ij}\) gives one linear equation.  
For a given column \(j\), the submatrix \(C_{\text{obs}}\) from the observed rows has size \(m_j \times K\); the local least‑squares problem is fully determined as soon as \(m_j \ge K\).  
Hence,  the whole matrix can be recovered \emph{exactly} (in the noiseless case), if every column of $\Phi$ receives at least \(K\) observations (half of the column entries).  
With a uniformly random mask of density \(p\), the expected number of observations per column is \(2Kp\); the probability that a column has \(<K\) entries falls rapidly only when \(p\) is well above \(\frac{1}{2}\).  
Thus, although the factorized method needs only \(O(K N)\) observations in principle, a purely random mask requires a substantial fraction (\(\approx 70\%\)–\(90\%\)) to guarantee that all columns are sufficiently populated.  
This behavior is visible in Fig.\ref{fig:recovery_plot}: the factorized recovery error remains high until the observation fraction becomes large and then drops sharply. On the plotted grid it does not reach machine precision, because a random mask can leave some columns with fewer than $K$ observed entries; these underdetermined columns dominate the remaining error.  
SVP, on the other hand, exploits correlations across columns and can impute missing entries, producing a smoother decrease in error as more data are available.

\paragraph{Noise resilience.}
Fig.\ref{fig:recovery_noise} examines the scenario where additive complex Gaussian noise of standard deviation \(\sigma\) contaminates the observed amplitudes. This is a noise model for estimated amplitudes, not a direct finite-shot simulation of computational-basis counts.  
Both methods use the same fixed observation fraction \(p=0.7\) (chosen so that most columns receive \(\ge K\) entries).  
SVP acts as a global low‑rank denoiser and suppresses noise more effectively, yielding a relative error approximately proportional to \(\sigma\).  
The factorized method solves independent least‑squares problems per column, and noise in a single column can directly perturb that column’s reconstruction; as a result, its error grows faster with \(\sigma\).  
For small noise levels the errors remain close to the noise-free floor set by the underdetermined columns. As \(\sigma\) increases, the factorized error grows faster than the SVP error; no end-to-end shot budget is inferred from this amplitude-noise model.

\begin{figure*}[t]
\centering
\begin{subfigure}[b]{0.48\textwidth}
\includegraphics[width=\linewidth]{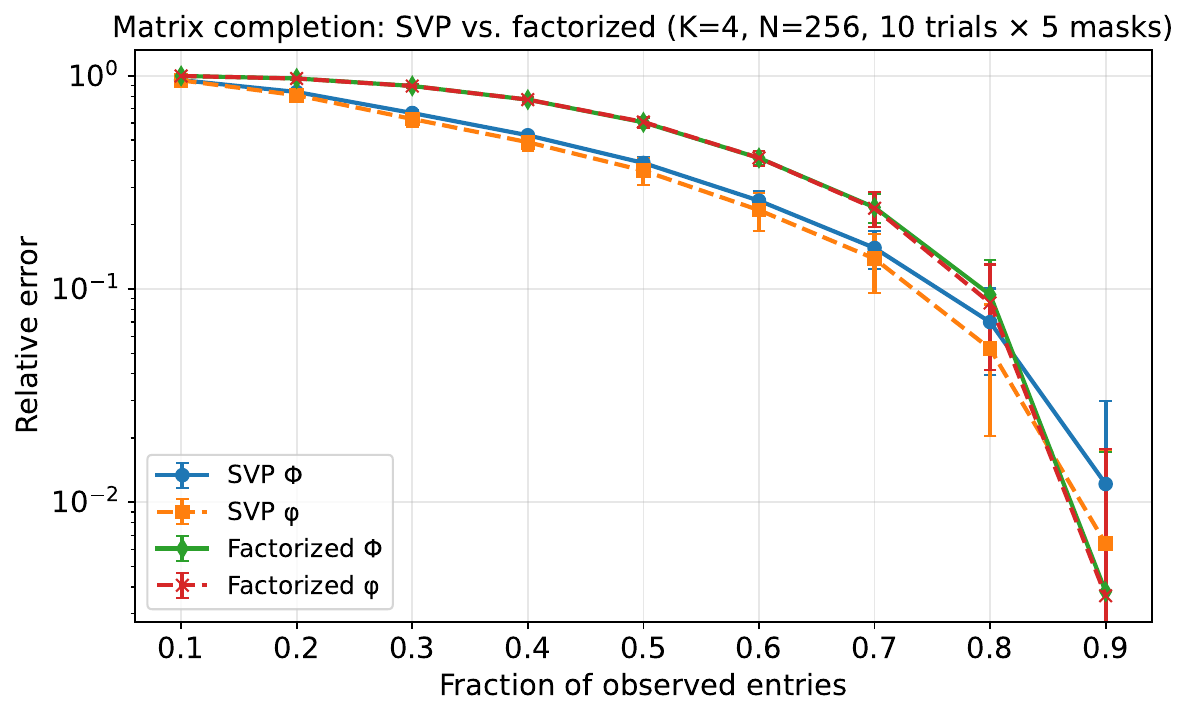}
\caption{\(N=256\), \(K=4\).}
\end{subfigure}\hfill
\begin{subfigure}[b]{0.48\textwidth}
\includegraphics[width=\linewidth]{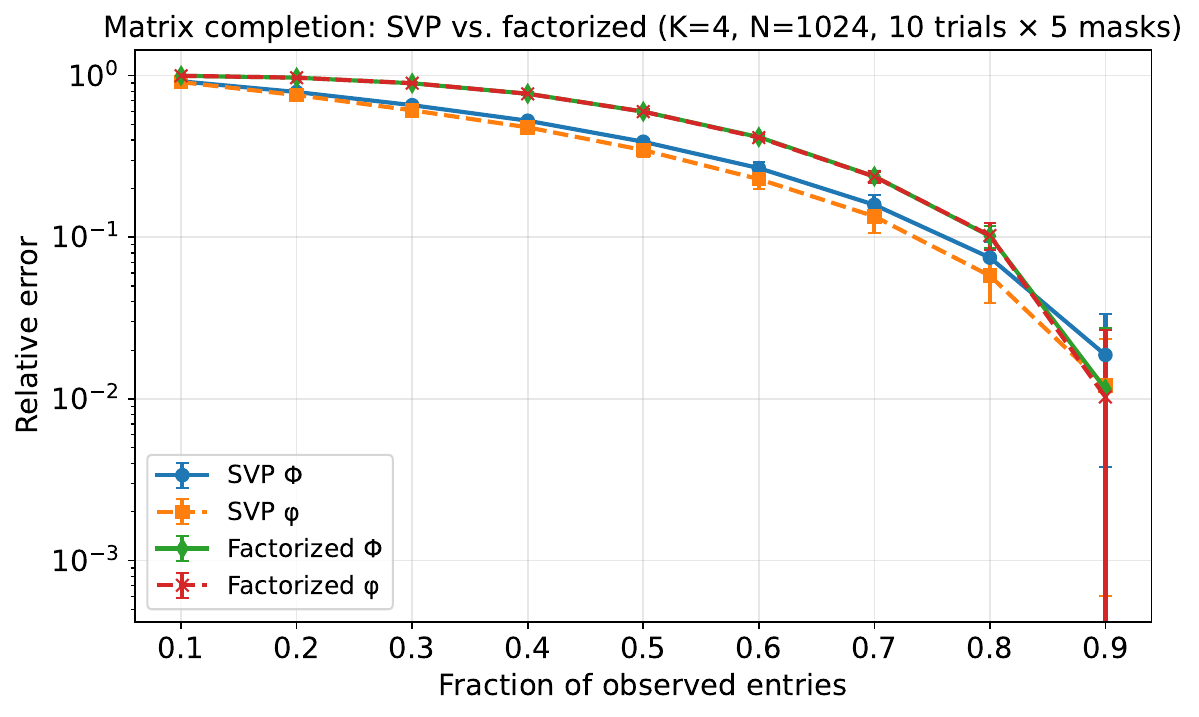}
\caption{\(N=1024\), \(K=4\).}
\end{subfigure}
\caption{Matrix completion of \(\Phi\) (solid) and \(\varphi\) (dashed) for exact amplitudes.  
SVP (blue) shows a gradual improvement, while factorized recovery (red) drops sharply as the observation fraction increases. Residual error can remain when some columns have fewer than \(K\) observed entries.  
Shaded bands: \(\pm1\) std.\ over \(10\) random instances and \(5\) masks.}
\label{fig:recovery_plot}
\end{figure*}

\begin{figure*}[t]
\centering
\begin{subfigure}[b]{0.48\textwidth}
\includegraphics[width=\linewidth]{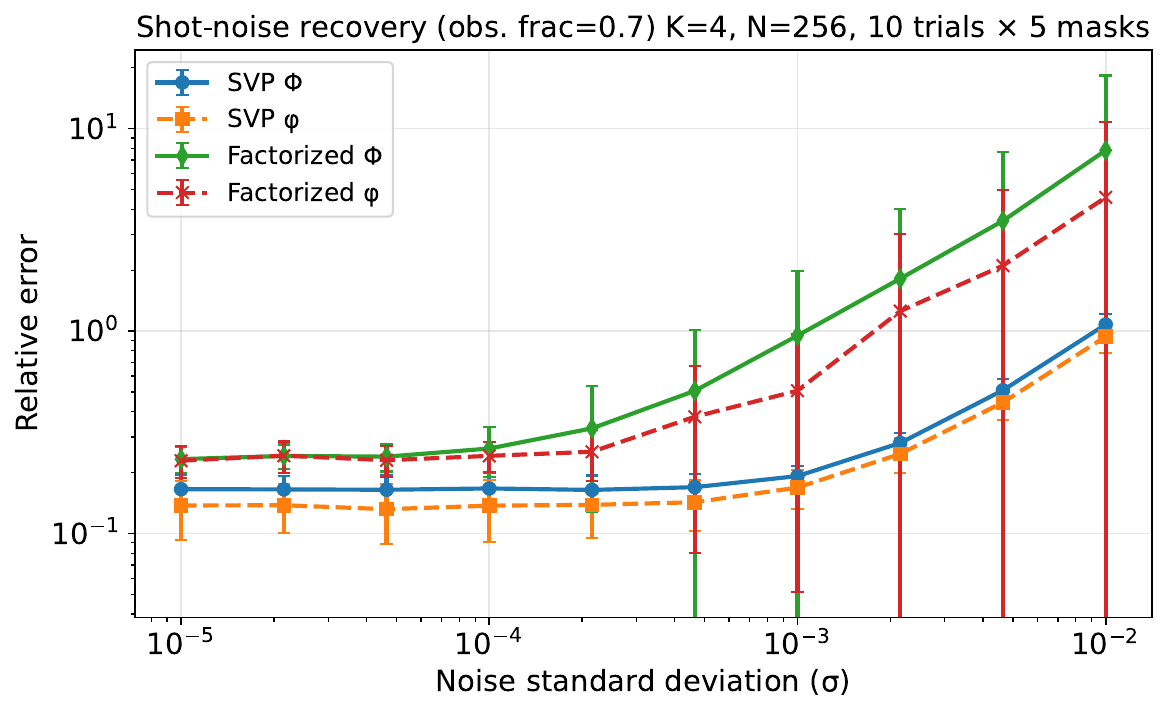}
\caption{\(N=256\), \(K=4\).}
\end{subfigure}\hfill
\begin{subfigure}[b]{0.48\textwidth}
\includegraphics[width=\linewidth]{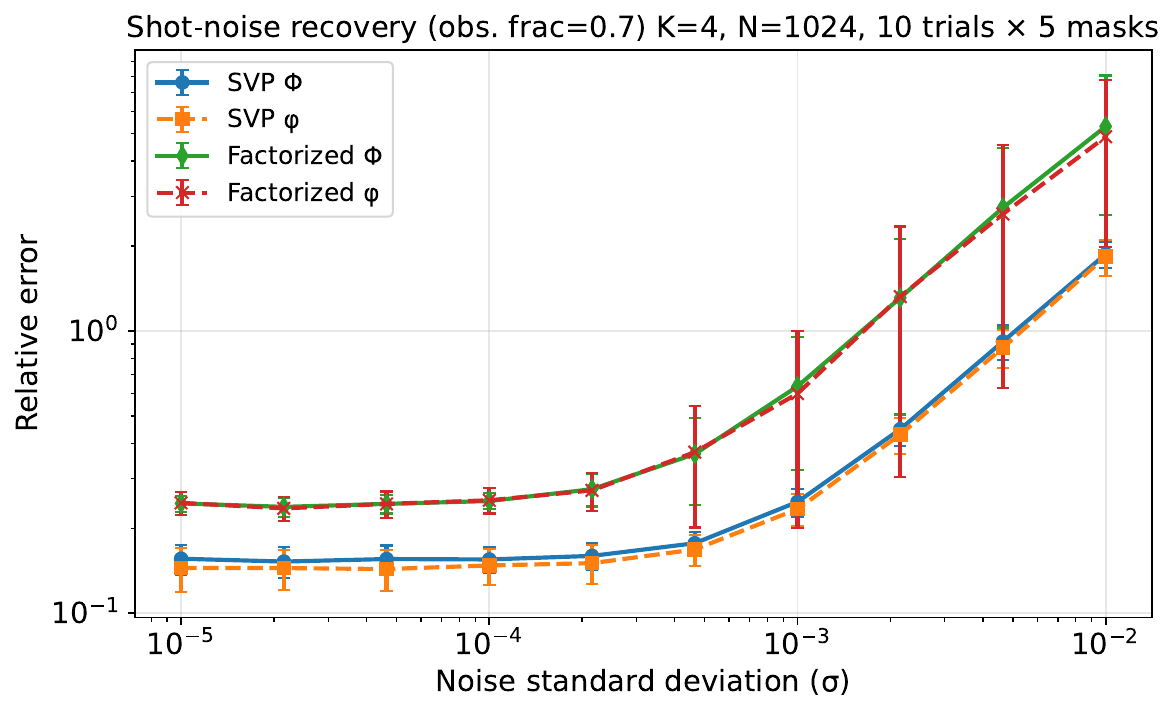}
\caption{\(N=1024\), \(K=4\).}
\end{subfigure}
\caption{Recovery under additive complex Gaussian noise with observation fraction fixed at 70\%.  
SVP (blue) denoises more effectively; factorized recovery (red) is more noise-sensitive. At small \(\sigma\), both remain near the noise-free floor determined by the observation mask.  
Shaded bands: \(\pm1\) std.\ over \(10\) instances and \(5\) noise realizations each.}
\label{fig:recovery_noise}
\end{figure*}

\section{Quantum trapdoor function from hidden coefficients}
\label{sec:trapdoor}

The factorization \(\Phi = C X\) derived in~\eqref{eq:factorization} not only enables structured recovery of the \(2K\times N\) output matrix in the amplitude-access model, but also suggests a cryptographic primitive: if the coefficients \(w_t\) are kept secret, the matrix \(C\) might act as a trapdoor key. Whether this hides the input state depends on the observer and the available measurements, and is not proved here.

This idea connects to a growing body of work on quantum trapdoor functions~\cite{brakerski2021cryptographic,coladangelo2023quantum}, quantum information hiding through classical keys~\cite{gheorghiu2022quantum,broadbent2015quantum}, and quantum money~\cite{bozzio2018experimental}, as well as to obfuscation of quantum circuits~\cite{alagic2016quantum,bartusek2024quantum}.
Moreover, families of quantum states can be computationally indistinguishable from Haar‑random~\cite{ji2018pseudorandom,brakerski2019pseudo}. Whether suitable random choices of the unitaries in our circuit produce a related pseudorandom ensemble is a separate question that we do not address.
The algebraic simplicity of this setting may nevertheless facilitate future security reductions based on structured phase retrieval, non‑negative matrix factorization~\cite{vavasis2010complexity}, or pseudorandom-state constructions~\cite{brakerski2019pseudo}.

We can describe a candidate quantum trapdoor function based directly on the LCU circuit of Fig.~\ref{fig:circuit}. With the notation of Section~\ref{sec:recovery}, the full output matrix $\Phi$ of size $2K\times N$ satisfies
\begin{equation}
\Phi = C \, X,
\qquad 
C_{(0,i),t} = \frac{1}{K}\, s_{i,t}\, w_t,\quad
C_{(1,i),t} = \frac{1}{K}\, s_{i,t}\, r_t,
\label{eq:Cdef_trap}
\end{equation}
where $X_{t,k} = \langle k|U_t|\psi\rangle$.
The matrix $C$ is $2K\times K$ with orthogonal columns of norm $1/\sqrt{K}$, hence it has full column rank $K$.
When $C$ is known, the unique solution for $X$ is $X = K\,C^\top\Phi$, from which the target state $T|\psi\rangle$ (and, by inverting the unitaries $U_t$, the original $|\psi\rangle$) can be recovered.

Assume now that the unitaries $U_t$ and the circuit structure (Hadamard gates, measurement pattern) are public, but the $K$ real coefficients $w_t$ (equivalently the rotation angles) are private.
An adversary who obtains $\Phi$—or an estimate thereof—faces the equation $\Phi = \tilde{C} \tilde{X}$ without knowing the true $C$.
Because $C$ is not square, the decomposition is not unique: for any invertible $K\times K$ matrix $M$, the pair $(\tilde{C},\tilde{X}) = (C M,\, M^{-1} X)$ is equally valid.
Without additional information the adversary cannot distinguish the true $X$ from $M^{-1}X$ using this factorization alone. This non-uniqueness motivates the candidate, but does not by itself establish that $|\psi\rangle$ is hidden, because the public circuit supplies additional structure.

The adversary does, however, possess side information: the public circuit description fixes $C$ to be of the form~\eqref{eq:Cdef_trap}, with unknown weights $w_t$ and known Hadamard sign pattern $s_{i,t}$.
Thus the problem reduces to recovering the $K$ real parameters $w_1,\dots,w_K$ from the product $\Phi = C(w)\,X$. In classical factorization, The hardness of recovering an unknown \(C\) from \(C X\) is known to be NP-hard in general  as evidenced by results on dictionary learning~\cite{spielman2012exact,vavasis2010complexity}. However, here since we construct  $C$ from the coefficients and signs of the Hadamard matrix, this makes the factorization easier given that the mixing unitary is public and the adversary has full tomographic information.

Below we analyze the security of the proposed trapdoor, point out this weakness that appears if the mixing unitary is public and the adversary has full tomographic information, and then show how the primitive can be strengthened—first by exploiting the computational‑basis measurement that naturally hides the phase, and then by moving to the standard LCU construction in which the mixing unitary itself depends on the secret key.

\subsection{Limitations of this basic trapdoor when complex amplitudes available}
\label{sec:trapdoor-basic}
To make the analysis easier, we can summarize the trapdoor function as follows:
\begin{itemize}
\item \textbf{Public parameters:} the \(K\) unitaries \(U_t\), the circuit structure,
  and the fact that the mixing unitary is the Hadamard matrix \(H\).
\item \textbf{Secret key:} the real coefficients \(w_1,\dots,w_K\)
  (with \(r_t=\sqrt{1-w_t^2}\)), i.e. the rotation angles.
\item \textbf{Evaluation:} prepare an input state \(\ket{\psi}\) on the system register,
  run the circuit, and measure the index, rotation, and system qubits in the
  computational basis. Repeat for many shots to build the empirical estimates
  \(\{\hat{p}_{i,r,k}\}\) of the probabilities \(p_{i,r,k}=|\Phi_{(i,r),k}|^2\).
  The output is the classical dataset \(\{\hat{p}_{i,r,k}\}\).
\item \textbf{Inversion with key:} if the full complex matrix \(\Phi\) is available
  through additional tomography, knowing \(C\) gives \(X = K\,C^{\!\top}\Phi\)
  and hence \(T\ket{\psi}\). From computational-basis intensities alone, inversion
  generally requires solving a phase-retrieval problem and is not guaranteed by
  the factorization alone.
\end{itemize}

\paragraph{An algebraic vulnerability – full complex information.}
If an adversary could obtain the exact complex amplitudes \(\Phi_{(i,r),k}\) (e.g.\ through
full state tomography on the ancilla‑system output), the trapdoor with a public
Hadamard mixing would be broken by a simple linear computation:
\begin{equation}
   \Phi_0 = \frac{1}{K} S\,\mathrm{diag}(\mathbf{w})\,X , \qquad
\Phi_1 = \frac{1}{K} S\,\mathrm{diag}(\mathbf{r})\,X , 
\end{equation}
where \(S_{i,t}=s_{i,t}=\sqrt{K}H_{it}\in\{\pm1\}\) is the (public) Hadamard‑sign matrix.
Multiplying by \(S^{\!\top}\) (which equals \(K\,S^{-1}\)) yields
\begin{equation}
S^{\!\top}\Phi_0 = \mathrm{diag}(\mathbf{w})\,X,\qquad
S^{\!\top}\Phi_1 = \mathrm{diag}(\mathbf{r})\,X .
\end{equation}
For any \(t\) with \(w_t\neq0\) and \(\mathbf{x}_t\neq 0\), the ratio of the two rows gives
\(r_t/w_t = \sqrt{1-w_t^2}/w_t\), from which \(w_t\) is uniquely recovered.
Thus, with full tomographic knowledge, the secret key is deterministically extracted
and the trapdoor fails.

\subsection{Hardness from phase‑retrieval: the realistic measurement model}
\label{sec:trapdoor-phase}

In the standard quantum computing framework and measurements, the evaluation procedure does not easily give the adversary the complex amplitudes \(\Phi_{ij}\). Instead, computational‑basis measurements reveal only the squared magnitudes \(p_{i,r,k}=|\Phi_{(i,r),k}|^2\). Therefore, the adversary’s input can be considered a noisy estimate of the probability vector \(\mathbf{p} = \{|\Phi_{ij}|^2\}\). Inverting the trapdoor without the key therefore becomes a structured phase‑retrieval problem \cite{candes2015phase,jagatap2017fast}: 
Given  \(\mathbf{p} = |C X|^2\) with  \(C\) known only to be of the  form  
\begin{equation}
    C = \frac{1}{K}\!\begin{pmatrix}S\,\mathrm{diag}(\mathbf{w})\\
  S\,\mathrm{diag}(\mathbf{r})\end{pmatrix},
\end{equation}
recover  \(\mathbf{w}\)  and \(X\) or directly \(T\ket{\psi}\). This phase retrieval, even in its unconstrained form, is a difficult non‑convex optimization \cite{candes2015phase,jagatap2017fast}. When the matrix is low‑rank (\(K\ll N\)) and its left factor is an unknown parametric dictionary, the problem becomes an instance of blind
dictionary learning under magnitude constraints \cite{tillmann2023extended}. Generic phase-retrieval and dictionary-learning problems can be computationally difficult, but their worst-case hardness does not by itself prove average-case hardness for this particular structured family.

Furthermore, different choices of auxiliary degrees of freedom, such as relative phases of the rows of \(X\), can sometimes produce the same squared magnitudes. Such ambiguities may obstruct inversion in particular cases, but they do not establish information-theoretic hiding for all inputs or observers.

Thus, the basic Hadamard‑based circuit, when used with computational‑basis measurements, is an informal candidate for a quantum trapdoor family. Any security claim would require an explicit average-case hardness assumption and a precise observer model. A generic secret mixing unitary $W$, as in standard LCU, may remove some of the public Hadamard structure, but it does not by itself resolve these issues.

\subsection{Strengthening construction by using generic coefficient matrix}
\label{sec:trapdoor-secret-mixing}

One may worry that the publicly known Hadamard mixing provides too much algebraic structure. A possible modification is to replace the fixed Hadamard matrix \(H\) by a secret generic mixing unitary \(W\in\mathrm{U}(K)\), generated from the secret key. The circuit is otherwise unchanged, and the coefficient matrix becomes
\begin{equation}
   C_{(0,i),t} = \frac{1}{\sqrt{K}}\,W_{it}\,w_t , \qquad
C_{(1,i),t} = \frac{1}{\sqrt{K}}\,W_{it}\,r_t .
\end{equation}
Here we use \(W\) only to hide the mixing and do not impose an additional first-row normalization; preserving a particular post-selected target would require a consistent choice of preparation and rotation weights.

Now the mixing matrix \(W\) itself is unknown to the adversary. Even if full
complex \(\Phi\) were available, extracting \(\mathbf{w}\) requires solving the
non‑linear system
\begin{equation}
    \Phi_0 = \frac{1}{\sqrt{K}}\,W\,\mathrm{diag}(\mathbf{w})\,X,\qquad
\Phi_1 = \frac{1}{\sqrt{K}}\,W\,\mathrm{diag}(\mathbf{r})\,X,
\end{equation}

with \(W\in\mathrm{U}(K)\). This is a structured
dictionary‑learning problem where the dictionary depends on the hidden weights.
The simple public-Hadamard separation attack no longer applies directly because \(W\) is not a fixed public matrix; other attacks are not ruled out.

At an informal level, the proposed one-wayness assumption is that no efficient adversary, given the public parameters and the measurement statistics, can recover \(\mathbf{w}\) or \(T\ket{\psi}\) for the input distribution of interest. The adversary receives only the squared magnitudes of the entries of \(\Phi=C(\mathbf{w},W)X\) and faces a structured phase-retrieval problem. A possible optimization formulation is
\begin{equation}
\operatorname{argmin}_{\tilde{\mathbf{w}},\tilde{X}}
\Bigl\|\, |C(\tilde{\mathbf{w}},\tilde{W})\,\tilde{X}|^2 - \hat{\mathbf{p}}\,\Bigr\|,
\end{equation}
subject to the public circuit constraints. This problem is non-convex, but non-convexity and ambiguity alone do not prove computational or information-theoretic security. A formal average-case hardness reduction, a precise observer model, and an analysis of known-input attacks are left for future work. Secret generic mixing should therefore be regarded only as a possible strengthening of the candidate, not as a proved trapdoor function.

\section{Discussion and conclusion}
\label{sec:discussion}

\paragraph{Circulant variant.}
Replacing the Hadamard matrix $H$ by the discrete Fourier transform matrix $F$ (entries $F_{jt} = \omega^{jt}/\sqrt{K}$ with $\omega = e^{2\pi i/K}$) turns the blocks $A$ and $B$ into block‑circulant matrices.
In this case the index register can be diagonalized by a quantum Fourier transform, and each Fourier mode yields exactly one of the scaled unitaries $w_t U_t$ (up to index relabeling).
The eigenvalues of $A$ then become the eigenvalues of the $w_t U_t$ blocks, giving full spectral control while keeping the whole construction unitary.
This variant trades the simple $\pm1$ Hadamard signs for an analytically solvable spectrum, and we believe it may be useful when the eigenvalue distribution of the linear combination is important.

\paragraph{Explicit circuit synthesis.}
The cosine‑sine decomposition (CSD) of the shuffled unitary $U$ (Section~\ref{sec:CSD}) expresses it in terms of the singular values of $A$ and $B$, which are directly given by the coefficients $w_t$ and $r_t$. They can be used as input parameters for a CSD‑based circuit synthesis algorithm. Such an approach could produce a fully structured quantum circuit that implements the exact linear combination with minimal gate overhead.

\paragraph{Potential applications.}
The factorization $\Phi = C X$ is exactly the low‑rank model that underpins collaborative‑filtering recommendation systems~\cite{koren2009matrix}: $C$ plays the role of a “user‑feature” matrix and $X$ the “item‑feature” matrix.
Thus the LCU circuit of Fig.\ref{fig:circuit} can be viewed as a quantum embedding of a rank‑$K$ factor model, where the coefficients $\alpha_t$ encode latent user preferences and the unitaries $U_t$ generate item profiles.
In a recommendation scenario one might prepare a superposition over user types, run the circuit, and explore recovery of the factor model in the amplitude-access model described above.
Similarly, PageRank and related link‑analysis algorithms~\cite{borodin2005link} iterate $T = \alpha P + (1-\alpha)\mathbf{1}\mathbf{v}^\top$, a two‑term LCU that fits naturally into the Hadamard‑mixing circuit; the output matrix $\Phi$ would then give the stationary distribution from every measurement branch.

Measuring the rotation qubit in the $|\pm\rangle$ basis extracts $A\pm B$, providing modified linear combinations $w_t\pm r_t$ that could simplify quantum circuits for matrix addition and subtraction. These directions are speculative, but they illustrate how the “junk” ancilla outcomes might be repurposed beyond standard LCU.

On the cryptographic side, the involution property $V^2 \propto \bigoplus_t U_t^2$ suggests a symmetric-circuit construction with identical forward and backward operations; if $U_t^2=I$, the circuit becomes an exact involution. By itself this gives neither confidentiality nor authentication.
In addition, the trapdoor construction given in this paper may open several research directions. From a practical standpoint, one could instantiate the unitaries \(\{U_t\}\) with random
permutations or with elements of a Pauli subgroup, making the circuit
implementation lightweight. From a theoretical angle, one may attempt to prove
security in the quantum random‑oracle model or establish a connection with the
pseudorandomness of the output states \cite{ji2018pseudorandom}. We plan to
explore these aspects in a separate, cryptography‑focused paper.

\paragraph{Conclusion.}
In this paper we consider exploiting the information hidden in the full state output of LCU circuits. We have shown that a small change to the standard LCU circuit-Hadamard mixing combined with a single rotation qubit-exposes a clean algebraic structure in all ancilla outcomes.
By reshaping the measurement results into a $2K\times N$ matrix, one obtains the factorization $\Phi = C X$ with rank at most $K$.
This observation suggests two application directions: (i) in an amplitude-access model, the full matrix can be recovered from partial observations using low‑rank completion, while computational-basis counts generally require additional phase information for full complex recovery; (ii) the coefficient matrix $C$ can be kept secret, motivating an informal candidate quantum trapdoor function. The same algebraic structure also links LCU to classical matrix factorization models and suggests connections to recommendation systems, ranking algorithms, and quantum linear algebra primitives.
We leave finite-shot recovery from measurement counts, formal security reductions, larger‑scale numerical benchmarks, and the integration with suggested algorithms as future work.

\section*{Acknowledgments}
The author acknowledges the use of different AI models for language refinement during the preparation of this manuscript.

\section*{Data Availability}
The simulation code is publicly available on GitHub at \url{https://github.com/adaskin/lcu-full-output-recovery}. All experimental results can be reproduced using the provided code. 

\section*{Conflict of Interest Statement}
Author declares no conflict of interest.
\section*{Funding}
This work is not supported by any funding agency.
\section*{Ethics Statement}
Not applicable.
\bibliographystyle{unsrt}
\bibliography{main}

\end{document}